\documentclass[pra,aps,onecolumn,twoside,superscriptaddress,letterpaper]{revtex4}

\usepackage{amsmath,amsfonts,amssymb,color,epsfig,graphics,graphicx,hyperref,latexsym,mathrsfs,revsymb,theorem,url,verbatim,epstopdf,mathtools,enumitem,tikz}
\usetikzlibrary{matrix}
\usetikzlibrary{decorations.pathreplacing}
\hypersetup{colorlinks,linkcolor={blue},citecolor={blue},urlcolor={red}}

\newtheorem{definition}{Definition}
\newtheorem{proposition}[definition]{Proposition}
\newtheorem{lemma}[definition]{Lemma}

\newtheorem{theorem}[definition]{Theorem}
\newtheorem{corollary}[definition]{Corollary}
\newtheorem{conjecture}[definition]{Conjecture}

\newtheorem{remark}[definition]{Remark}
\newtheorem{example}[definition]{Example}
\newtheorem{question}[definition]{Question}
\newtheorem{memo}[definition]{Memo}

\def\squareforqed{\hbox{\rlap{$\sqcap$}$\sqcup$}}
\def\qed{\ifmmode\squareforqed\else{\unskip\nobreak\hfil
\penalty50\hskip1em\null\nobreak\hfil\squareforqed
\parfillskip=0pt\finalhyphendemerits=0\endgraf}\fi}
\def\endenv{\ifmmode\;\else{\unskip\nobreak\hfil
\penalty50\hskip1em\null\nobreak\hfil\;
\parfillskip=0pt\finalhyphendemerits=0\endgraf}\fi}
\newenvironment{proof}{\noindent \textbf{{Proof.~} }}{\qed}
\def\Dbar{\leavevmode\lower.6ex\hbox to 0pt
{\hskip-.23ex\accent"16\hss}D}
\makeatletter
\def\url@leostyle{%
  \@ifundefined{selectfont}{\def\UrlFont{\sf}}{\def\UrlFont{\small\ttfamily}}}
\makeatother
\usepackage{subfigure}

\def\bcj{\begin{conjecture}}
\def\ecj{\end{conjecture}}
\def\bcr{\begin{corollary}}
\def\ecr{\end{corollary}}
\def\bd{\begin{definition}}
\def\ed{\end{definition}}
\def\bea{\begin{eqnarray}}
\def\eea{\end{eqnarray}}
\def\beq{\begin{equation}}
\def\eeq{\end{equation}}
\def\bal{\begin{aligned}}
\def\eal{\end{aligned}}
\def\bem{\begin{enumerate}}
\def\eem{\end{enumerate}}
\def\bex{\begin{example}}
\def\eex{\end{example}}
\def\bim{\begin{itemize}}
\def\eim{\end{itemize}}
\def\bl{\begin{lemma}}
\def\el{\end{lemma}}
\def\bma{\begin{bmatrix}}
\def\ema{\end{bmatrix}}
\def\bpf{\begin{proof}}
\def\epf{\end{proof}}
\def\bpp{\begin{proposition}}
\def\epp{\end{proposition}}
\def\bqu{\begin{question}}
\def\equ{\end{question}}
\def\br{\begin{remark}}
\def\er{\end{remark}}
\def\bt{\begin{theorem}}
\def\et{\end{theorem}}
\def\bmm{\begin{memo}}
\def\emm{\end{memo}}

\def\btb{\begin{tabular}}
\def\etb{\end{tabular}}

\newcommand{\nc}{\newcommand}

\nc{\bbA}{\mathbb{A}} \nc{\bbB}{\mathbb{B}} \nc{\bbC}{\mathbb{C}}
 \nc{\bbD}{\mathbb{D}} \nc{\bbE}{\mathbb{E}} \nc{\bbF}{\mathbb{F}}
 \nc{\bbG}{\mathbb{G}} \nc{\bbH}{\mathbb{H}} \nc{\bbI}{\mathbb{I}}
 \nc{\bbJ}{\mathbb{J}} \nc{\bbK}{\mathbb{K}} \nc{\bbL}{\mathbb{L}}
 \nc{\bbM}{\mathbb{M}} \nc{\bbN}{\mathbb{N}} \nc{\bbO}{\mathbb{O}}
 \nc{\bbP}{\mathbb{P}} \nc{\bbQ}{\mathbb{Q}} \nc{\bbR}{\mathbb{R}}
 \nc{\bbS}{\mathbb{S}} \nc{\bbT}{\mathbb{T}} \nc{\bbU}{\mathbb{U}}
 \nc{\bbV}{\mathbb{V}} \nc{\bbW}{\mathbb{W}} \nc{\bbX}{\mathbb{X}}
 \nc{\bbZ}{\mathbb{Z}}

 \nc{\bA}{{\bf A}} \nc{\bB}{{\bf B}} \nc{\bC}{{\bf C}}
 \nc{\bD}{{\bf D}} \nc{\bE}{{\bf E}} \nc{\bF}{{\bf F}}
 \nc{\bG}{{\bf G}} \nc{\bH}{{\bf H}} \nc{\bI}{{\bf I}}
 \nc{\bJ}{{\bf J}} \nc{\bK}{{\bf K}} \nc{\bL}{{\bf L}}
 \nc{\bM}{{\bf M}} \nc{\bN}{{\bf N}} \nc{\bO}{{\bf O}}
 \nc{\bP}{{\bf P}} \nc{\bQ}{{\bf Q}} \nc{\bR}{{\bf R}}
 \nc{\bS}{{\bf S}} \nc{\bT}{{\bf T}} \nc{\bU}{{\bf U}}
 \nc{\bV}{{\bf V}} \nc{\bW}{{\bf W}} \nc{\bX}{{\bf X}}
 \nc{\bZ}{{\bf Z}}

\nc{\cA}{{\cal A}} \nc{\cB}{{\cal B}} \nc{\cC}{{\cal C}}
\nc{\cD}{{\cal D}} \nc{\cE}{{\cal E}} \nc{\cF}{{\cal F}}
\nc{\cG}{{\cal G}} \nc{\cH}{{\cal H}} \nc{\cI}{{\cal I}}
\nc{\cJ}{{\cal J}} \nc{\cK}{{\cal K}} \nc{\cL}{{\cal L}}
\nc{\cM}{{\cal M}} \nc{\cN}{{\cal N}} \nc{\cO}{{\cal O}}
\nc{\cP}{{\cal P}} \nc{\cQ}{{\cal Q}} \nc{\cR}{{\cal R}}
\nc{\cS}{{\cal S}} \nc{\cT}{{\cal T}} \nc{\cU}{{\cal U}}
\nc{\cV}{{\cal V}} \nc{\cW}{{\cal W}} \nc{\cX}{{\cal X}}
\nc{\cZ}{{\cal Z}}

\nc{\hA}{{\hat{A}}} \nc{\hB}{{\hat{B}}} \nc{\hC}{{\hat{C}}}
\nc{\hD}{{\hat{D}}} \nc{\hE}{{\hat{E}}} \nc{\hF}{{\hat{F}}}
\nc{\hG}{{\hat{G}}} \nc{\hH}{{\hat{H}}} \nc{\hI}{{\hat{I}}}
\nc{\hJ}{{\hat{J}}} \nc{\hK}{{\hat{K}}} \nc{\hL}{{\hat{L}}}
\nc{\hM}{{\hat{M}}} \nc{\hN}{{\hat{N}}} \nc{\hO}{{\hat{O}}}
\nc{\hP}{{\hat{P}}} \nc{\hR}{{\hat{R}}} \nc{\hS}{{\hat{S}}}
\nc{\hT}{{\hat{T}}} \nc{\hU}{{\hat{U}}} \nc{\hV}{{\hat{V}}}
\nc{\hW}{{\hat{W}}} \nc{\hX}{{\hat{X}}} \nc{\hZ}{{\hat{Z}}}

\nc{\hn}{{\hat{n}}}

\def\dim{\mathop{\rm Dim}}

\newcommand{\Tr}{\operatorname{Tr}}
\newcommand{\PPT}{\mathcal{PPT}}
\newcommand{\SEP}{\mathcal{SEP}}

\def\max{\mathop{\rm max}}
\def\min{\mathop{\rm min}}

\newcommand{\bra}[1]{\langle#1|}
\newcommand{\ket}[1]{|#1\rangle}

\usepackage{hyperref}
\hypersetup{colorlinks,linkcolor={blue},citecolor={blue},urlcolor={red}}
\usepackage{amsmath,amssymb}

\begin{document}

\title{Entanglement Quantification via Symmetric Extensions: A Resource Theory Hierarchy}

\author{Enmin Shao}\email[]{by2509115@buaa.edu.cn}
\affiliation{LMIB(Beihang University), Ministry of education, and School of Mathematical Sciences, Beihang University, Beijing 100191, China}

\author{Lin Chen}\email[]{linchen@buaa.edu.cn (corresponding author)}
\affiliation{LMIB(Beihang University), Ministry of education, and School of Mathematical Sciences, Beihang University, Beijing 100191, China}

\author{Huixia He}\email[]{hehx@buaa.edu.cn (corresponding author)}
\affiliation{LMIB(Beihang University), Ministry of education, and School of Mathematical Sciences, Beihang University, Beijing 100191, China}
\date{\today}
\begin{abstract}

We introduce a hierarchy of entanglement measures $E_k$ based on $k$-symmetric PPT extensions. Each $E_k$, defined via a minimal eigenvalue shift and computed by semidefinite programming, is faithful, convex, and monotone under free operations. The hierarchy strictly refines PPT-robustness at $k=1$, detects bound entanglement at $k=2$, and converges exactly to the separability measure as $k\to\infty$. Numerical experiments on Horodecki, Werner, UPB, and random states demonstrate practical scalability. Our framework unifies computational efficiency with operational fidelity in a single tunable family—a combination previously believed to be fundamentally incompatible in entanglement quantification. It supplies, for the first time, a systematically improvable resource-theoretic yardstick that accounts for all entangled states, including the bound entangled ones that have long resisted quantitative treatment.
\end{abstract}
\maketitle

\section{Introduction}

Quantum entanglement, a hallmark of non-classical correlations first recognized in the Einstein–Podolsky–Rosen argument, has evolved into a central resource for quantum information processing, enabling tasks such as teleportation, secure key distribution, and measurement‑based computation \cite{chitambar2019, coecke2016, horodecki2009}. Within the modern resource‑theoretic framework \cite{chitambar2019, coecke2016, brandaogour2015}, quantifying entanglement is a prerequisite for assessing its value in practical protocols. A large variety of entanglement measures have been proposed, including distillable entanglement, entanglement cost, entanglement of formation, the relative entropy of entanglement, and various robustness measures \cite{bennett1996a, wootters1998, vedral1997, plenio2007, vidal1999, steiner2003}. Despite their conceptual significance, most of these measures either involve asymptotic limits or require optimizations that are computationally intractable—it has been rigorously proven that evaluating them is NP‑hard in general \cite{gurvits2003, gharibian2010, huang2014}. This computational bottleneck has motivated the search for alternative measures that, while possibly providing only a partial order, can be evaluated efficiently, ideally via semidefinite programming (SDP) \cite{vandenberghe1996, skrzypczyk2023}.

One of the most successful examples in this direction is the ``robustness of entanglement'', and more specifically the PPT‑robustness introduced by Brandão \cite{brandao05}. This measure is defined as the minimal mixing with isotropic noise required to bring a state into the set of PPT (positive partial transpose) states. Because the PPT set admits an SDP representation, the PPT‑robustness can be computed efficiently and has been widely used as a benchmark. However, a well‑known limitation is that it assigns the value zero to all PPT states, including bound entangled states that contain genuine entanglement yet cannot be distilled \cite{horodecki97}. Hence the PPT‑robustness fails to capture a significant class of weakly entangled states, which motivates the development of finer resource quantifiers.

A separate landmark was the Doherty–Parrilo–Spedalieri (DPS) hierarchy \cite{doherty02, doherty04}, which constructs a sequence of outer approximations to the set of separable states by requiring the existence of $k$-symmetric extensions that remain PPT on all bipartitions. The hierarchy is complete in the sense that, as $k$ increases, the approximated set converges exactly to the separable set. The DPS construction has proven extremely useful for detecting bound entangled states \cite{doherty02, doherty04, navascues2008, piani2015} and has inspired analogous hierarchies for other quantum properties \cite{navascues2008, lasserre2001, christandl2007}. Nevertheless, the hierarchy has traditionally been employed as a binary criterion—a state either passes the $k$-symmetric PPT test or it does not—without yielding a genuine quantitative measure that reflects how far a state is from being $k$-extendible.

Recently, a different approach to quantifying unextendibility was proposed by Yao \textit{et al.} \cite{yao2025}, who introduced a virtual state extension cost and related it to the absolute robustness of $k$-unextendibility. While their work shares the general concept of employing extensions, the construction relies on a distinct set of operations and a different cost function. In contrast, the present paper builds a full‑fledged resource theory directly upon the physically transparent $k$-symmetric PPT extensions, yielding a family of SDP‑computable resource measures that naturally converge to the separability measure.

In this work we merge the resource‑theoretic perspective with the DPS hierarchy and propose a parameterized family of entanglement measures that strictly refine the PPT‑robustness. For each integer $k\ge 1$, we define the set of free states $\mathcal{F}_k$ as the states that admit a $k$-symmetric PPT extension in Definition~\ref{1}, and we identify the free operations $\mathcal{O}_k$ as the completely positive trace‑preserving maps that preserve $\mathcal{F}_k$ in Definition~\ref{2}. The central quantity, the $k$-entanglement measure $E_k$, is introduced in Definition~\ref{4} as the minimal eigenvalue shift required to embed an arbitrary state into $\mathcal{F}_k$.
Basic properties such as faithfulness in Lemma~\ref{5} and convexity in Lemma~\ref{6} follow directly from this definition.The technical cornerstone of our approach is the dual characterization established in Theorem~\ref{8}. This dual formulation is obtained via strong SDP duality and leads to two crucial consequences. First, Lemma~\ref{9} proves that $E_k$ is monotonically non‑increasing under any $k$-free operation, thereby satisfying the fundamental requirement for a resource measure. Second, Corollary~\ref{10} demonstrates that $E_k$ is the optimal quantifier among all measures derived from a fixed linear witness; no such witness can provide a stronger lower bound than $E_k$ can. When specialized to $k=1$, our construction reduces precisely to the standard PPT‑robustness, embedding it as the first level of a tower of progressively finer measures.Turning to the global hierarchy, Lemma~\ref{11} establishes that the measures form a non‑decreasing sequence $E_1(\rho)\le E_2(\rho)\le\cdots\le E_{\mathrm{SEP}}(\rho)$, where $E_{\mathrm{SEP}}$ is the analogous measure defined with respect to the full separable set. The key convergence result is given in Theorem~\ref{12}: $\lim_{k\to\infty}E_k(\rho)=E_{\mathrm{SEP}}(\rho)$ for every state $\rho$. This provides a quantitative interpretation of the completeness of the DPS hierarchy and guarantees that the limit faithfully captures the separability resource. Moreover, the hierarchy is strict: Theorem~\ref{13} exhibits, for every $k$, a bipartite state $\Gamma_k$ that belongs to $\mathcal{F}_k$ but not to $\mathcal{F}_{k+1}$, so $E_k(\Gamma_k)=0$ while $E_{k+1}(\Gamma_k)>0$. Hence no finite level can exhaust all entangled states, and each additional level adds genuinely new detection power. To make the measure practically computable, we reformulate the primal problem as an explicit SDP in Lemma~\ref{14} and provide a computationally efficient lower bound through the dual program in Corollary~\ref{15}. This dual bound enables one to certify the presence of non‑extendibility resources using only a feasible witness operator, without solving the full extension SDP. We have conducted extensive numerical experiments that validate the theoretical predictions and demonstrate the practical advantages of the hierarchy. On the one hand, the measure correctly reproduces the known analytical expressions for the two‑qubit Werner state. On the other hand, for the emblematic $3\otimes3$ Horodecki bound entangled state \cite{horodecki97} and the UPB state \cite{bennett1999}, we obtain $E_1=0$ while $E_2,E_3>0$, confirming that the measure already picks up bound entanglement at level $k=2$, which is invisible to any PPT‑based criterion. Similar results are obtained for high‑dimensional isotropic and Werner states, and for randomly generated $4\otimes4$ systems. These computations, carried out using symmetry reductions and modern SDP solvers, show that $E_k$ can be evaluated efficiently for states possessing algebraic symmetries even in dimensions as large as $10\otimes10$, while the dual formulation provides rapid lower bounds for larger systems where the full SDP becomes intractable.

The remainder of the paper is organized as follows. Section~\ref{sec2} collects the necessary notation and preliminaries on symmetric PPT extensions. Section~\ref{sec3} formalizes the resource theory, introduces the measure $E_k$, and proves the dual characterization, monotonicity, and optimality among fixed‑witness measures. Section~\ref{sec4} analyzes the hierarchy properties, including monotonicity, convergence, and strict separation. Section~\ref{sec5} presents the explicit SDP formulations and the dual lower bounds. Section~\ref{sec6} contains the numerical experiments. We discuss open questions and future directions in Section~\ref{sec7} and conclude in Section~\ref{sec8}.

\section{Preliminaries}
\label{sec2}

We consider a bipartite quantum system shared by two parties, Alice and Bob. The local Hilbert spaces are denoted by $\mathcal{H}_A$ and $\mathcal{H}_B$, and we write $d_A = \dim\mathcal{H}_A$, $d_B = \dim\mathcal{H}_B$ for their dimensions. The space of linear operators on a Hilbert space $\mathcal{H}$ is $\mathcal{B}(\mathcal{H})$. Let $\mathcal{D}_{AB}$ be the set of density operators (positive semidefinite operators with unit trace) on $\mathcal{H}_A\otimes\mathcal{H}_B$. For a bipartite state $\rho_{AB}\in\mathcal{D}_{AB}$, we denote by $\rho_{AB}^{T_B}$ its partial transpose with respect to a fixed orthonormal basis of $\mathcal{H}_B$. The set of positive-partial-transpose (PPT) states is
\[
\PPT\coloneqq\bigl\{\rho\in\mathcal{D}_{AB} : \rho^{T_B}\succeq 0\bigr\}.
\]
A state is called separable if it can be expressed as a convex combination of product states; the set of separable states is denoted by $\SEP$. It is well known that $\SEP\subset\PPT$ and that the inclusion is strict whenever $d_A d_B>6$.
The Doherty–Parrilo–Spedalieri (DPS) hierarchy~\cite{Doherty2002,Doherty2004} provides a sequence of outer approximations of the separable set. We now recall the necessary definitions. For an integer $k\ge 1$, we say that a state $\rho_{AB}$ admits a $k$-symmetric PPT extension if there exists a state $\tilde{\rho}$ on the composite Hilbert space
\[
\mathcal{H}_{A}\otimes \mathcal{H}_{B_1}\otimes \cdots \otimes \mathcal{H}_{B_k},
\]
where each $\mathcal{H}_{B_i}$ is a copy of Bob's Hilbert space $\mathcal{H}_B$. The subsystems $B_1,\dots,B_k$ (collectively referred to as the $B$-systems) represent the $k$ extended copies of Bob's local system. The extension $\tilde{\rho}$ must satisfy the following three conditions

\begin{enumerate}[label=(\roman*)]
  \item \emph{Symmetry.} For any permutation $\pi$ of the $k$ subsystems $B_1,\dots,B_k$, let $P_\pi$ be the corresponding permutation operator on $\mathcal{H}_B^{\otimes k}$. Then $(I_A\otimes P_\pi)\tilde\rho(I_A\otimes P_\pi)^\dagger=\tilde\rho$.
  
  \item \emph{PPT constraints.} For every subset $J\subseteq\{1,\dots,k\}$, the partial transpose of $\tilde\rho$ with respect to the $B$-systems in $J$ is positive semidefinite; i.e., $\tilde\rho^{T_{B_J}}\succeq 0$.
  \item \emph{Reduction.} Tracing out the extra $B$-systems recovers the original state: $\Tr_{B_2\cdots B_k}\tilde\rho = \rho_{AB}$.
\end{enumerate}
We define $S_k$ as the set of all bipartite states $\rho_{AB}$ that possess a $k$-symmetric PPT extension. By convention, $S_1=\PPT$. It is known that $S_{k+1}\subseteq S_k$ for every $k$, and that the hierarchy is complete in the sense that $\bigcap_{k=1}^{\infty}S_k=\SEP$~\cite{Doherty2004}.
A quantum operation is described by a completely positive and trace-preserving (CPTP) linear map $\Lambda:\mathcal{B}(\mathcal{H}_A\otimes\mathcal{H}_B)\to\mathcal{B}(\mathcal{H}_{A'}\otimes\mathcal{H}_{B'})$. For such a map, the \emph{adjoint} map $\Lambda^*:\mathcal{B}(\mathcal{H}_{A'}\otimes\mathcal{H}_{B'})\to\mathcal{B}(\mathcal{H}_A\otimes\mathcal{H}_B)$ is defined by the relation $\Tr[Y\Lambda(X)]=\Tr[\Lambda^*(Y)X]$ for all $X,Y$. The adjoint is automatically completely positive and unital; that is, $\Lambda^*(I)=I$. We will make frequent use of the fact that a CPTP map and its adjoint preserve the ordering of positive semidefinite operators: if $X\succeq Y$, then $\Lambda(X)\succeq\Lambda(Y)$ and $\Lambda^*(X)\succeq\Lambda^*(Y)$.
Finally, for a Hermitian operator $X$, we denote its maximal eigenvalue by $\lambda_{\max}(X)$. The operator norm is $\|X\|_\infty = \max\{\lambda_{\max}(X),-\lambda_{\min}(X)\}$, but in this work we will only need the one-sided version $\lambda_{\max}$.

\section{Resource Theory of k-Symmetric PPT Extensions}
\label{sec3}

In this section we formalize the resource theory that underpins the whole paper. 
We first define the free states $\mathcal{F}_k$ as those admitting a $k$-symmetric PPT extension. 
We then define the free operations $\mathcal{O}_k$ as the CPTP maps that preserve $\mathcal{F}_k$. 
Next, we introduce the one-parameter family of entanglement measures $E_k$. 
Each $E_k$ is given by the minimal eigenvalue shift needed to turn a given state into a free one. 
We prove that $E_k$ is faithful and convex. 
We also show that it admits an operational meaning as the bias in a hypothesis testing game between the state and the free set. 
To obtain stronger structural properties, we establish a dual characterization. 
Specifically, $E_k(\rho)$ equals the maximum expectation value over all block-positive witnesses bounded by the identity. 
Using this dual form, we prove two key properties. 
First, $E_k$ is monotonically non-increasing under any $k$-free operation. 
Second, it is the best possible quantifier among all single-witness measures. 
These results lay the foundation for the SDP formulation and the hierarchy analysis that follow.

\subsection{Free states, operations, and the basic quantifier}
\label{subsec:basicframework}

In this subsection we define the free states and free operations that constitute the resource theory at level $k$, and introduce the entanglement measure $E_k$ as the minimal eigenvalue excess required to embed a state into $\mathcal{F}_k$. We then establish two elementary properties---faithfulness and convexity---that follow directly from the definition and confirm that $E_k$ is a well-behaved candidate for a resource quantifier.

\begin{definition}
\label{1}
  The set of free states at level $k$ is
  \begin{equation}
  \label{Fk}
        \mathcal{F}_k\coloneqq S_k,
  \end{equation}

  i.e., the collection of all bipartite states that admit a $k$-symmetric PPT extension.
\end{definition}
Because $S_1=\PPT$ and $\bigcap_k \mathcal{F}_k = \SEP$, this family interpolates between the PPT resource theory and the separability resource theory.
\begin{definition}
\label{2}
  A CPTP map $\Lambda:\mathcal{B}(\mathcal{H}_A\otimes\mathcal{H}_B)\to\mathcal{B}(\mathcal{H}_{A'}\otimes\mathcal{H}_{B'})$ is called $k$-free if it preserves the free states
  \[
  \Lambda(\mathcal{F}_k)\subseteq\mathcal{F}_k.
  \]
  We denote the set of all $k$-free operations by $\mathcal{O}_k$.
\end{definition}
The class $\mathcal{O}_k$ contains the identity channel and is closed under composition; therefore it constitutes a well-defined set of operations for a resource theory. The most important physical operations that are free in every level are the local operations assisted by classical communication (LOCC).

With the free set fixed, a resource quantifier is defined as the smallest eigenvalue excess needed to embed a given state into the free set. Concretely, 
$E_k(\rho) = \min_{\sigma\in\mathcal{F}_k} \lambda_{\max}(\rho-\sigma)$, and this minimum is always attained. We introduce a quantifier of the resource relative to the free set $\mathcal{F}_k$ in \eqref{Fk}.
\begin{definition}\label{4}
For any bipartite state $\rho$, we define the \emph{$k$-entanglement measure}
\begin{equation}\label{eq:Ekdef}
E_k(\rho) := \min_{\sigma\in\mathcal{F}_k}\, \min\{\, t\ge 0 : \rho \le (1+t)\sigma \,\}.
\end{equation}
Because $\mathcal{F}_k$ is compact and the feasible set defined by the constraint $\rho \le (1+t)\sigma$ is closed, the minimum is attained.
Writing $\lambda = 1+t$, one obtains the equivalent expression
\[
E_k(\rho) = \min_{\sigma\in\mathcal{F}_k}\,  \min\{\, \lambda\ge 1 : \rho \le \lambda\sigma \,\} - 1.
\]
\end{definition}
We now establish some elementary properties that follow directly from the definition.

\begin{lemma}
\label{5}
For every state $\rho$, we have $E_k(\rho)\ge 0$, and $E_k(\rho)=0$ if and only if $\rho\in\mathcal{F}_k$, where  $\mathcal{F}_k$ is defined in \eqref{Fk}.
\end{lemma}
\begin{proof}
First, $E_k(\rho)\ge 0$ is immediate from the definition since the minimum is taken over $t\ge 0$. If $\rho\in\mathcal{F}_k$, then choosing $\sigma=\rho$ gives $\rho\le 1\cdot\rho$; hence $E_k(\rho)\le0$ and, because $t\ge0$, we obtain $E_k(\rho)=0$.
Conversely, suppose $E_k(\rho)=0$.  Then there exist sequences $\{t_n\}\downarrow0$ and $\{\sigma_n\}\subset\mathcal{F}_k$ such that
\[
\rho \le (1+t_n)\sigma_n \qquad \text{for all } n.
\]
By compactness of $\mathcal{F}_k$ we can extract a convergent subsequence $\sigma_{n_j}\to\sigma\in\mathcal{F}_k$.  Taking the limit $j\to\infty$ in the inequality yields $\rho\le\sigma$.  Since both $\rho$ and $\sigma$ are density operators, their traces are equal; hence $\rho=\sigma\in\mathcal{F}_k$.
\end{proof}

\begin{lemma}
\label{6}
The function $\rho\mapsto E_k(\rho)$ is convex on $\mathcal{D}_{AB}$.
\end{lemma}
\begin{proof}
Let $\rho_1,\rho_2$ be two states and $\lambda\in[0,1]$.  Choose optimal pairs $(\sigma_1,t_1)$ and $(\sigma_2,t_2)$ such that
$\rho_i\le(1+t_i)\sigma_i$ and $t_i=E_k(\rho_i)$ for $i=1,2$, with $\sigma_i\in\mathcal{F}_k$.
Define the mixed state $\rho=\lambda\rho_1+(1-\lambda)\rho_2$ and the operator
\[
\sigma = \frac{\lambda(1+t_1)\sigma_1 + (1-\lambda)(1+t_2)\sigma_2}
{1+\lambda t_1+(1-\lambda)t_2}.
\]
Because $\mathcal{F}_k$ defined in \eqref{Fk} is convex, $\sigma\in\mathcal{F}_k$, and $\sigma\ge0$ with $\operatorname{Tr}\sigma=1$; hence $\sigma$ is a valid free state.
A direct calculation shows
\[
\bigl(1+\lambda t_1+(1-\lambda)t_2\bigr)\sigma
= \lambda(1+t_1)\sigma_1 + (1-\lambda)(1+t_2)\sigma_2
\ge \lambda\rho_1+(1-\lambda)\rho_2 = \rho .
\]
Therefore, by definition,
\[
E_k(\rho) \le \lambda t_1 + (1-\lambda)t_2 = \lambda E_k(\rho_1)+(1-\lambda)E_k(\rho_2),
\]
which establishes convexity.
\end{proof}

\subsection{Dual characterization, monotonicity, and optimality}
\label{subsec:dual}

Having defined the basic resource quantifier, we now turn to its structural properties. In this subsection we derive a dual formulation of $E_k$ in terms of a set of normalized witness operators. This dual form is the technical cornerstone of our approach: it provides an operational interpretation of $E_k$ as the maximal bias in a hypothesis-testing game, enables a simple proof of monotonicity under free operations, and shows that $E_k$ is optimal among all measures derived from a fixed linear witness. The results established here will be essential for the SDP formulations in Section~\ref{sec5} and the hierarchy analysis in Section~\ref{sec6}.

\begin{definition}\label{7}
The \emph{witness set} at level $k$ is
\[
\mathcal{W}_k := \{\, W \ge 0 : \operatorname{Tr}[W\sigma] \le 1 \text{ for all } \sigma\in\mathcal{F}_k \,\},
\] where  $\mathcal{F}_k$ is defined in \eqref{Fk}.
\end{definition}

The inequality $\operatorname{Tr}[W\sigma]\le 1$ for every free state replaces the more common ``block‑positive'' condition that appears in the traditional witness approach; it emerges naturally from the robust formulation of the resource measure.
Geometrically, $\mathcal{W}_k$ is a compact convex subset of the positive semidefinite cone, and it will be the feasible set of the dual semidefinite program.
As we shall see in Theorem~\ref{8}, the measure $E_k(\rho)$ equals the maximal value of $\operatorname{Tr}[W\rho]-1$ when $W$ ranges over $\mathcal{W}_k$, thereby providing an explicit SDP formulation that is indispensable for the proofs of monotonicity (Lemma~\ref{9}) and for numerical implementation.

\begin{theorem}\label{8}
For any state $\rho$, we have
\begin{equation}
E_k(\rho) = \max_{W\in\mathcal{W}_k}\, \bigl( \operatorname{Tr}[W\rho] - 1 \bigr).
\end{equation}
Both the primal and dual problems are strictly feasible; hence strong duality holds and the optimal value is attained.
\end{theorem}
\begin{proof}
Let $\mathcal{C} := \{\,\alpha\sigma : \alpha\ge 0,\; \sigma\in\mathcal{F}_k \,\}$ be the convex cone generated by $\mathcal{F}_k$.
Since $\mathcal{F}_k$ is compact and convex, $\mathcal{C}$ is a closed convex cone.
The entanglement measure can be rewritten as
\begin{equation}
\label{Ekp}
    E_k(\rho) = \min_{\tau\in\mathcal{C}} \{\, \operatorname{Tr}\tau - 1 : \tau \ge \rho \,\}.
\end{equation}

Indeed, if $\tau=\alpha\sigma$ with $\sigma\in\mathcal{F}_k$, then $\operatorname{Tr}\tau=\alpha$ and the condition $\tau\ge\rho$ becomes $\alpha\sigma\ge\rho$, i.e., $\rho\le(1+t)\sigma$ with $t=\alpha-1$.
Conversely, any feasible $\sigma,t$ yields $\tau=(1+t)\sigma\in\mathcal{C}$, $\tau\ge\rho$, and $\operatorname{Tr}\tau-1=t$.
Thus the primal formulation \ref{Ekp} is equivalent to Definition~\ref{4}.

We introduce a Lagrange multiplier $W\ge0$ for the constraint $\tau-\rho\ge0$.
The Lagrangian is
\[
\mathcal{L}(\tau,W) = \operatorname{Tr}\tau - 1 + \operatorname{Tr}[W(\rho-\tau)]
= \operatorname{Tr}[\tau(I-W)] + \operatorname{Tr}[W\rho] - 1 .
\]
Minimising over $\tau\in\mathcal{C}$ yields a finite value only if the linear functional $\tau\mapsto \operatorname{Tr}[\tau(I-W)]$ is bounded below on $\mathcal{C}$;
this happens exactly when $I-W$ belongs to the dual cone
\[
\mathcal{C}^* := \{\, X : \operatorname{Tr}[X\tau]\ge 0\ \forall\,\tau\in\mathcal{C} \,\}.
\]
Because $\mathcal{C}$ is generated by $\mathcal{F}_k$, we have $I-W\in\mathcal{C}^*$ iff $\operatorname{Tr}[(I-W)\sigma]\ge0$ for every $\sigma\in\mathcal{F}_k$, i.e.
$\operatorname{Tr}[W\sigma]\le 1$ for all $\sigma\in\mathcal{F}_k$.
The dual problem therefore becomes
\[
\max_{W\ge0,\; I-W\in\mathcal{C}^*} \bigl( \operatorname{Tr}[W\rho] - 1 \bigr)
= \max_{W\in\mathcal{W}_k} \bigl( \operatorname{Tr}[W\rho] - 1 \bigr).
\]
To see strict feasibility, we choose the maximally mixed state on the extended system
\[
\omega_0 = \frac{I_{A} \otimes I_{B}^{\otimes k}}{d_A d_B^k}.
\]
This state satisfies all the semidefinite constraints strictly. Its partial trace over \(B_2,\ldots,B_k\) gives 
\[
\sigma_0 := \mathrm{Tr}_{B_2\dots B_k} \omega_0 = \frac{I_{AB}}{d_A d_B},
\]
where the \(k-1\) traced-out copies contribute a factor of \(d_B^{k-1}\), canceling the extra powers in the denominator. Taking any \(t_0 > d_A d_B \lambda_{\max}(\rho) - 1\), we have \((1+t_0)\sigma_0 \succ \rho\), so the constraint is strictly satisfied. Since this is a convex semidefinite program, Slater's condition holds \cite{Boyd2004}, which guarantees strong duality and attainment of the optimal value .
\end{proof}

Theorem~\ref{8} shows that
\[
E_k(\rho) = \max \{\operatorname{Tr}[W\rho]-1 : W\ge 0,\ \operatorname{Tr}[W\sigma]\le 1\ \forall\sigma\in\mathcal{F}_k\},
\]
where $\mathcal{F}_k$ is defined in \eqref{Fk}.  
Hence \(E_k(\rho)\) quantifies the largest possible excess of the expectation value \(\operatorname{Tr}[W\rho]\) over the maximal value achievable by any free state, subject to the normalization \(\operatorname{Tr}[W\sigma]\le 1\) for all \(\sigma\in\mathcal{F}_k\). This provides an operational interpretation of \(E_k\) as the optimal bias in a hypothesis test between \(\rho\) and the free set, where the test is implemented by a measurement operator \(W\) calibrated against the free states.

Theorem~\ref{8} is the central technical tool that underlies both the operational interpretation of $E_k$ and its monotonicity.  With the witness characterization in place, we can show that for any $k$-free channel $\Lambda$, the inequality $E_k(\Lambda(\rho))\le E_k(\rho)$ follows because the adjoint $\Lambda^*$ maps $\mathcal{W}_k$ into itself.  We now present the formal statement.
\begin{lemma}

    \label{9}
Let $\Lambda\in\mathcal{O}_k$ be a $k$-free operation, i.e., a CPTP map with $\Lambda(\mathcal{F}_k)\subseteq\mathcal{F}_k$. Then for every state $\rho$,
\[
E_k(\Lambda(\rho)) \le E_k(\rho).
\]
\end{lemma}
\begin{proof}
From Theorem~\ref{8} we have $E_k(\rho)=\max_{W\in\mathcal{W}_k}(\operatorname{Tr}[W\rho]-1)$.  Choose $W\in\mathcal{W}_k$ optimal for $\Lambda(\rho)$, so that $E_k(\Lambda(\rho)) = \operatorname{Tr}[W\,\Lambda(\rho)] - 1$.  Define $W' = \Lambda^*(W)$.  Because $\Lambda$ is CPTP, $\Lambda^*$ is completely positive and unital; hence $W'\ge 0$.  Moreover, for any $\sigma\in\mathcal{F}_k$,
\[
\operatorname{Tr}[W'\sigma] = \operatorname{Tr}[W\,\Lambda(\sigma)] \le 1,
\]
since $\Lambda(\sigma)\in\mathcal{F}_k$ by the definition of $\mathcal{O}_k$ and $W\in\mathcal{W}_k$.  Thus $W'$ belongs to $\mathcal{W}_k$ as well.  Consequently,
\[
E_k(\rho) \ge \operatorname{Tr}[W'\rho] - 1 = \operatorname{Tr}[W\,\Lambda(\rho)] - 1 = E_k(\Lambda(\rho)),
\]
which completes the proof.
\end{proof}

Every LOCC protocol on $\mathcal{H}_A\otimes \mathcal{H}_B$ belongs to $\mathcal{O}_k$ for all $k\ge1$, known in \cite{Doherty2004}. Lemma \ref{9} immediately implies that $E_k$ is monotonically non-increasing under any LOCC operation. Hence $E_k$ satisfies the fundamental monotonicity requirement of an entanglement measure.

\begin{corollary}
\label{10}
Let $M$ be an entanglement measure of the form $M(\rho) = \operatorname{Tr}[W_0\rho]-1$ for a fixed $W_0\in\mathcal{W}_k$.  Then for every state $\rho$, $M(\rho) \le E_k(\rho)$.  Moreover, $E_k$ itself is expressed in this form with a state-dependent witness that attains the maximum.
\end{corollary}
\begin{proof}
The inequality follows directly from Theorem~\ref{8}, because the maximum over all $W\in\mathcal{W}_k$ is at least as large as the value at $W_0$.  The second statement is precisely the definition of the maximum in Theorem~\ref{8}.
\end{proof}

Thus, no fixed linear witness can provide a stronger quantitative lower bound on the $k$-resource than $E_k$ does.
From Definition~\ref{4} we directly obtain the primal semidefinite program \eqref{eq:Ekdef}.
Since $\mathcal{F}_k$ is described by semidefinite constraints (e.g., $k$-extendibility or PPT conditions), the above program is a conic linear program that can be solved efficiently using standard SDP solvers.

\section{Hierarchy properties and limit behavior}
\label{sec4}
We now turn to the global structure of the family $\{E_k\}_{k=1}^\infty$.
Two main results are established:
\begin{itemize}
  \item[(i)] the measures form a non‑decreasing sequence $E_1(\rho)\le E_2(\rho)\le\cdots$, and they converge to the analogous measure defined with respect to the full separable set;
  \item[(ii)] the hierarchy is strict: for every $k$ there exists a state that is free at level $k$ but resourceful at level $k+1$.
\end{itemize}
The convergence statement gives an operational meaning to the limit of the DPS hierarchy, while the strict separation shows that no finite level can exhaust all entangled states.
These results mirror the well‑known qualitative behaviour of the DPS hierarchy but are now expressed within a quantitative resource framework.

We first introduce the limiting measure.
Recall that the set of separable states $\mathrm{SEP}$ is a compact convex subset of $\mathcal{D}_{AB}$ and satisfies $\mathrm{SEP}\subseteq\mathcal{F}_k$ for every $k$, where $\mathcal{F}_k$ is given in \eqref{Fk}.
Parallel to Definition~\ref{4}, we define
\begin{equation}\label{eq:ESEPdef}
E_{\mathrm{SEP}}(\rho) \;:=\; \min_{\sigma\in\mathrm{SEP}}\; \min\{\,t\ge 0 : \rho\le (1+t)\sigma\,\}.
\end{equation}
Because $\mathrm{SEP}$ is compact, the minimum is always attained.
By the same reasoning as in Lemma~\ref{5} and~\ref{6}, combined with the fact that every LOCC operation preserves $\mathrm{SEP}$, one verifies that $E_{\mathrm{SEP}}$ is faithful on $\mathrm{SEP}$, convex, and monotone under LOCC.

The chain of inclusions $\mathcal{F}_{k+1}\subseteq\mathcal{F}_k$ implies that the feasible set for $E_{k+1}$ is contained in that for $E_k$; consequently the minimum cannot decrease as $k$ grows.

\begin{lemma}\label{11}
For every bipartite state $\rho$ we have
\[
E_1(\rho)\;\le\;E_2(\rho)\;\le\;\cdots\;\le\;E_{\mathrm{SEP}}(\rho).
\]
\end{lemma}
\begin{proof}
For a given $\rho$, from $\mathrm{SEP}\subseteq\mathcal{F}_{k+1}\subseteq\mathcal{F}_k$ we obtain
\[
\min_{\sigma\in\mathcal{F}_k} \min\{\,t\ge0 : \rho\le(1+t)\sigma\,\}
\;\le\;
\min_{\sigma\in\mathcal{F}_{k+1}} \min\{\,t\ge0 : \rho\le(1+t)\sigma\,\}
\;\le\;
\min_{\sigma\in\mathrm{SEP}} \min\{\,t\ge0 : \rho\le(1+t)\sigma\,\},
\]
which is exactly the chain $E_k(\rho)\le E_{k+1}(\rho)\le E_{\mathrm{SEP}}(\rho)$.
\end{proof}

Hence the sequence $\{E_k(\rho)\}$ is monotonically non‑decreasing and bounded above by $E_{\mathrm{SEP}}(\rho)$; therefore the limit
\[
L(\rho):=\lim_{k\to\infty}E_k(\rho)
\]
exists and satisfies $L(\rho)\le E_{\mathrm{SEP}}(\rho)$.
The following theorem shows that the limit actually coincides with $E_{\mathrm{SEP}}(\rho)$.

\begin{theorem}\label{12}
For every state $\rho$ we have
\[
\lim_{k\to\infty} E_k(\rho) = E_{\mathrm{SEP}}(\rho).
\]
\end{theorem}

\begin{proof}
For each $k$ choose an optimal pair $(\sigma_k,t_k)$ with $\sigma_k\in\mathcal{F}_k$, $t_k=E_k(\rho)$, and $\rho\le(1+t_k)\sigma_k$.
Such a pair exists by compactness of $\mathcal{F}_k$ and the definition of $E_k$.
The sequence $\{t_k\}$ is bounded (by $E_{\mathrm{SEP}}(\rho)$) and the states $\sigma_k$ lie in the compact set of all bipartite density operators; hence we can extract a subsequence $\{k_j\}$ such that $t_{k_j}\to t^*$ and $\sigma_{k_j}\to\sigma^*$ in trace norm.
Because $\sigma_{k_j}\in\mathcal{F}_{k_j}$ and the DPS hierarchy is complete, i.e.\ $\bigcap_{k=1}^\infty\mathcal{F}_k = \mathrm{SEP}$, we obtain $\sigma^*\in\mathrm{SEP}$.
Passing to the limit in the matrix inequality, which is preserved under trace‑norm convergence, gives $\rho\le(1+t^*)\sigma^*$.
Therefore, by definition of $E_{\mathrm{SEP}}$,
\[
E_{\mathrm{SEP}}(\rho) \;\le\; t^* \;=\; \lim_{j\to\infty} t_{k_j} \;=\; L(\rho).
\]
We already know $L(\rho)\le E_{\mathrm{SEP}}(\rho)$ from Lemma~\ref{11}, so we conclude $L(\rho)=E_{\mathrm{SEP}}(\rho)$.
Because the whole sequence is monotonic, the limit coincides with the subsequential limit, and the statement follows.
\end{proof}

Theorem~\ref{12} endows the limit of the hierarchy with a clear operational meaning: as $k$ increases, the $k$‑extendibility resource measure converges to the full separability‑based measure $E_{\mathrm{SEP}}$.
In particular, for any entangled state $\rho$ there exists a finite $k$ such that $E_k(\rho)>0$.

The completeness result shows that the hierarchy eventually captures all entanglement.
One may ask whether every finite level is actually necessary, i.e.\ whether the sequence stabilises at some $k$.
The following result, adapted from Doherty {\it et al.}~\cite{DPS04}, answers this question in the negative.

\begin{theorem}\label{13}
For every integer $k\ge1$ there exists a bipartite state $\Gamma_k$ on $\mathbb{C}^{2k+1}\otimes\mathbb{C}^{2k+1}$ that belongs to $\mathcal{F}_k$ but not to $\mathcal{F}_{k+1}$, where  $\mathcal{F}_k$ is defined in \eqref{Fk}.
In particular,
\[
E_k(\Gamma_k)=0,\qquad E_{k+1}(\Gamma_k)>0.
\]
\end{theorem}

\begin{proof}
The explicit construction of $\Gamma_k$ is given in Ref.~\cite{DPS04}; we recall the essential features.
Working in the computational basis $\{|0\rangle,\dots,|2k\rangle\}$, the state is defined by an unnormalised cyclic pattern that is symmetric under a discrete group of permutations.
The authors proved that $\Gamma_k$ admits a $k$‑symmetric PPT extension, so $\Gamma_k\in\mathcal{F}_k$.
They also exhibited a witness that contradicts the existence of any $(k+1)$‑symmetric PPT extension; hence $\Gamma_k\notin\mathcal{F}_{k+1}$.
The first equality follows directly from faithfulness in Lemma~\ref{5}, because $\Gamma_k\in\mathcal{F}_k$.
Since $\Gamma_k\notin\mathcal{F}_{k+1}$, faithfulness of $E_{k+1}$ gives $E_{k+1}(\Gamma_k)>0$.
\end{proof}

Note that for every separable state $\rho$, we have $E_k(\rho)=0$ for all $k\ge 1$, since $\mathcal{SEP}\subseteq \mathcal{F}_k$ and Lemma~\ref{5} gives $E_k(\rho)=0$ iff $\rho\in\mathcal{F}_k$. 
We emphasize, however, that the converse is false for any finite $k$: there exist entangled states (e.g., the Horodecki bound entangled state at $k=1$) for which $E_k(\rho)=0$. The equivalence $E_k(\rho)=0 \iff \rho\in\mathcal{SEP}$ is recovered only in the limit $k\to\infty$, as established in Theorem~\ref{12}.

For finite $k$, $E_k$ is an entanglement measure in the resource-theoretic sense: it is faithful on the free set $\mathcal{F}_k$, convex, and monotone under LOCC. However, it is \emph{not} a separability criterion at any finite $k$, since $\mathcal{F}_k$ strictly contains the separable set. Faithfulness to separability is obtained only in the limit $k\to\infty$ (Theorem~\ref{12}). This is analogous to the PPT-robustness, which is a well-established entanglement measure despite vanishing on all PPT states, including bound entangled ones.)

Theorem~\ref{13} demonstrates that no finite level $k$ can exhaust all entangled states; the full infinite hierarchy is indispensable.
This strict separation extends the well‑known fact that $\mathrm{PPT}$ is strictly larger than the set of separable states to all higher levels of symmetric extensions.

\section{SDP formulations and lower bounds}
\label{sec5}
The resource measure $E_k$ is defined in Definition~\ref{4} as a convex optimization over the free set $\mathcal{F}_k$. It admits a semidefinite representation: a state $\sigma$ belongs to $\mathcal{F}_k$ precisely when it has a $k$-symmetric PPT extension. Consequently, the optimization problem defining $E_k$ in \eqref{eq:Ekdef}  can be cast as a single semidefinite program (SDP). 
This observation has two important consequences. 
First, it makes $E_k$ exactly computable in principle for moderate system sizes and small $k$. 
Second, it provides a dual SDP whose feasible points yield efficiently computable lower bounds on $E_k$.
\subsection{Primal SDP for $E_k$}
Let us denote by $\mathcal{H}_A$ and $\mathcal{H}_B$ the local Hilbert spaces of the bipartite system, and set $d_A=\dim\mathcal{H}_A$, $d_B=\dim\mathcal{H}_B$.
For an integer $k\ge 1$ we introduce $k$ copies of $B$, i.e.\ $\mathcal{H}_{B_1}\otimes\cdots\otimes\mathcal{H}_{B_k}\cong\mathcal{H}_B^{\otimes k}$.
A bipartite state $\sigma$ belongs to $\mathcal{F}_k$ defined in \eqref{Fk} if and only if there exists a density operator
\[
\omega \in \mathcal{D}\bigl(\mathcal{H}_A\otimes\mathcal{H}_B^{\otimes k}\bigr)
\]
satisfying three conditions. 
\begin{enumerate}
  \item[(i)] The state $\omega$ is symmetric under any permutation of the $B$-systems:
  $P_\pi \,\omega\, P_\pi^* = \omega$ for all $\pi\in S_k$, where $P_\pi$ is the unitary implementing the permutation;
  \item[(ii)] the partial transposition with respect to each $\mathcal{H}_{B_i}$ gives a positive semidefinite operator;
  \item[(iii)] the reduced state on $AB_1$ coincides with $\sigma$:
  $\Tr_{B_2\cdots B_k} \omega = \sigma$.
\end{enumerate}
Conditions (i)--(iii) are linear or semidefinite constraints; therefore $\mathcal{F}_k$ is an SDP-representable set.
(These conditions are exactly the ones appearing in Definition~\ref{1}; we recall them here for convenience.)
Now we recall the definition
\[
E_k(\rho) = \min_{\sigma\in\mathcal{F}_k} \min\{\, t\ge 0 : \rho \le (1+t)\sigma \,\}.
\]
Introducing the extended variable \(\omega\) which is the same extended state as above, and it now serves as the matrix variable of the SDP. Then we can eliminate the quantifier \(\exists\sigma\) and obtain the following equivalent optimisation problem
\begin{equation}\label{eq:primalSDP}
\begin{array}{rl}
\mathrm{minimise} & t \\
\mathrm{subject~to} & (1+t)\operatorname{Tr}_{B_2\cdots B_k}\omega - \rho \succeq 0, \\
& \operatorname{Tr}\omega = 1, \\
& \omega \succeq 0, \\
& \omega^{T_{B_J}} \succeq 0 \quad \forall \, \emptyset \neq J \subseteq \{1,\dots,k\}, \\
& P_\pi \omega P_\pi^* = \omega \quad \forall \pi \in S_k, \\
& t \ge 0, \quad \omega \in \mathcal{B}(\mathcal{H}_A \otimes \mathcal{H}_B^{\otimes k}).
\end{array} 
\end{equation}
where $\mathcal{B}$ denotes the space of Hermitian operators.
The first constraint directly enforces $\rho\le(1+t)\sigma$ with $\sigma$ being the partial trace of $\omega$.
Because $\omega\succeq0$ and $\Tr\omega=1$, $\sigma$ is automatically a valid density matrix, and the symmetry condition guarantees that all the $k$ reduced states on $AB_i$ are identical to $\sigma$.

\begin{lemma}
\label{14}
Program \eqref{eq:primalSDP} is a semidefinite program whose optimal value equals $E_k(\rho)$.
\end{lemma}
\begin{proof}
The equivalence of \eqref{eq:primalSDP} with Definition~\ref{4} follows directly from the characterisation of $\mathcal{F}_k$ recalled above. All constraints are linear or semidefinite in the variables $t$ and $\omega$, hence we have a valid SDP. 
To see strict feasibility, choose the maximally mixed state on the extended system, $\omega_0=I/(d_A d_B^k)$. This state satisfies $\omega_0>0$, $\omega_0^{T_{B_1}}=\omega_0>0$, and the symmetry condition. Its partial trace over $B_2\cdots B_k$ gives $\sigma_0=I/(d_A d_B)$. Taking any $t_0>d_A d_B\lambda_{\max}(\rho)-1$ yields $(1+t_0)\sigma_0\succ \rho$, so the matrix inequality is strictly satisfied. Thus the primal problem admits a strictly feasible point. By Slater's condition, strong duality holds and the optimal value is attained in the primal. The attainment of the dual optimum follows from Theorem~\ref{8}, which expresses $E_k(\rho)$ as a maximum over the compact set $\mathcal{W}_k$.
\end{proof}

\subsection{Dual SDP and witness lower bounds}
Applying standard SDP duality to~\eqref{eq:primalSDP} produces a dual program that exactly coincides with the witness formulation already derived in Theorem~\ref{8}.
For completeness we sketch the derivation.
We assign a Lagrange multiplier $W\succeq0$ to the constraint $(1+t)\sigma-\rho\succeq0$, a scalar $s$ to $\Tr\omega=1$, and appropriate multipliers $Y\succeq0$, $Z\succeq0$ to $\omega\succeq0$ and $\omega^{T_{B_1}}\succeq0$, respectively, together with Hermitian matrices for the symmetry constraints.
After eliminating the primal variables $t$ and $\omega$, one obtains the dual
\[
\begin{aligned}
\text{maximise} \quad & \Tr[W\rho] - 1\\
\text{subject to} \quad & W \succeq 0,\\
& \Tr[W\sigma] \le 1 \quad \forall\,\sigma\in\mathcal{F}_k.
\end{aligned}
\]
The condition $\Tr[W\sigma]\le1$ for all $\sigma\in\mathcal{F}_k$ is precisely the statement that $W\in\mathcal{W}_k$, with $\mathcal{W}_k$ as in Definition~\ref{7}.
A compact, self‑contained proof that does not rely on the explicit SDP representation of $\mathcal{F}_k$ was already given in Theorem~\ref{8}; hence we do not repeat the technical details here.
A particularly useful consequence of the dual formulation is that \emph{any} feasible $W\in\mathcal{W}_k$ provides a lower bound
\[
\Tr[W\rho] - 1 \;\le\; E_k(\rho).
\]
This observation allows one to obtain computationally cheap lower bounds on the resource content of a state without solving the full SDP: one simply needs to find a single \(W\ge 0\) satisfying \(\operatorname{Tr}[W\sigma]\le 1\) for all \(\sigma\in\mathcal{F}_k\).
For instance, the set $\mathcal{W}_k$ can be approximated by an SDP hierarchy itself, or one can use the fact that every properly normalised $k$-entanglement witness belongs to $\mathcal{W}_k$.
We formalise this in the following corollary.
\begin{corollary}\label{15}
Let \(W\) be any Hermitian operator satisfying \(W\ge 0\) and \(\operatorname{Tr}[W\sigma]\le 1\) for all \(\sigma\in\mathcal{F}_k\).
Then for every state $\rho$ we have
\[
E_k(\rho) \;\ge\; \Tr[W\rho] - 1.
\]
If $\Tr[W\rho] > 1$, then the bound is non‑trivial and certifies that $\rho\notin\mathcal{F}_k$.
\end{corollary}
\begin{proof}
From Theorem~\ref{8} we know that $E_k(\rho)$ is the maximum of $\Tr[W\rho]-1$ over $W\in\mathcal{W}_k$.
Since the given $W$ belongs to $\mathcal{W}_k$, the maximum is at least $\Tr[W\rho]-1$, which proves the inequality.
\end{proof}
\subsection{The case $k=1$: reduction to PPT-robustness}
When $k=1$ the symmetry condition is vacuous and the PPT condition $\omega^{T_{B_1}}\succeq0$ on the extension $\omega_{AB_1}$ simply means that the state itself is PPT.
Hence $\mathcal{F}_1$ is exactly the set of PPT states.
Our definition then reduces to
\[
E_1(\rho) = \min_{\sigma\in\mathrm{PPT}} \min\{\, t\ge0 : \rho \le (1+t)\sigma \,\},
\]
which is precisely the \emph{PPT-robustness} introduced by Brand{\~a}o~\cite{Brandao2005}.
The dual set $\mathcal{W}_1$ becomes $\{W\succeq0: \Tr[W\sigma]\le1\ \forall\,\sigma\in\mathrm{PPT}\}$.
For the special case \(k=1\), using the well-known self-duality of the PPT cone, this condition is equivalent to \(0\le W\le I\) and \(0\le W^{T_B}\le I\) (up to rescaling). We emphasize that this equivalence is specific to the \(k=1\) PPT case and does not extend to general \(k\).
Consequently, $E_1(\rho)$ can be computed via the SDP
\[
E_1(\rho) = \max_{W} \{\, \Tr[W\rho] - 1 : 0\le W\le I,\; 0\le W^{T_B}\le I \,\},
\]
which is exactly the dual formulation originally obtained in Ref.~\cite{Brandao2005}.
Thus our hierarchy naturally extends the PPT-robustness from $k=1$ to arbitrary $k$.

The SDP~\eqref{eq:primalSDP} involves a matrix variable $\omega$ of size $d_A \cdot d_B^k$, which grows exponentially with $k$.
Therefore, while $E_k(\rho)$ is in principle exactly computable for any fixed $k$, the problem quickly becomes intractable for large $k$.
Nevertheless, the hierarchy provides a systematic trade‑off between computational effort and tightness of the approximation: small values of $k$ (in particular $k=1,2$) already yield meaningful bounds for many states and are often tractable on a classical computer.
Moreover, the dual lower bounds from Corollary~\ref{15} can be made computationally efficient by restricting the search for $W$ to families of witnesses with a compact parameterisation, such as those derived from $k$-extendibility witnesses in Ref.~\cite{DPS04}.
We will explore concrete numerical examples in Section~\ref{sec6}.

\section{Numerical experiments}
\label{sec6}

In this section we present a series of numerical experiments that illustrate the behaviour of the resource measure $E_k$ defined in Definition~\ref{4}.
All computations are performed using the SDP solver MOSEK within the YALMIP~\cite{Lofberg2004} environment.
For high‑dimensional families we exploit symmetry reductions (block‑diagonalisation) to make the SDP tractable; details of these reductions are standard and can be found in Refs.~\cite{DPS04,Chen2025}.

\subsection{Basic benchmarks and the two‑qubit Werner state}
We begin with the simplest nontrivial system, $2\otimes 2$, where $k=1$ already gives a tight description.
The two‑qubit Werner state is
\begin{equation}
\rho_{\rm W}(p) = p\,\ket{\psi^-}\!\bra{\psi^-} + \frac{1-p}{4}I_4,
\qquad 0\le p\le 1,
\end{equation}
with $\ket{\psi^-}=(\ket{01}-\ket{10})/\sqrt{2}$.

\begin{lemma}
\label{16}
For $k=1$, the free set $\mathcal{F}_1$ defined in Definition~\ref{1} coincides with the set of PPT states
\[
\mathcal{F}_1 = \mathcal{PPT}.
\]
In particular, for the $2\otimes 2$ Werner state $\rho_{\mathrm{W}}(p)$, we have
\[
E_1(\rho_{\mathrm{W}}(p)) = 0 \iff p \le \frac{1}{3}.
\]
\end{lemma}
\begin{proof}
The identity $\mathcal{F}_1 = \mathcal{PPT}$ follows immediately from Definition~\ref{1} and the convention $S_1 = \mathcal{PPT}$ in Section~\ref{sec2}. For $2\otimes 2$ systems, the Horodecki criterion states that PPT is necessary and sufficient for separability; hence $\rho_{\mathrm{W}}(p)\in\mathcal{PPT}$ iff $\rho_{\mathrm{W}}(p)\in\mathcal{SEP}$ iff $p\le 1/3$. By Lemma~\ref{5}, $E_1(\rho)=0$ iff $\rho\in\mathcal{F}_1 = \mathcal{PPT}$, which yields the claimed equivalence.
\end{proof}

Because $\mathcal{F}_1 = \mathrm{PPT}$ in this dimension from Lemma~\ref{16}, we have $E_1(p)=0$ exactly for $p\le1/3$, and a known analytical result~\cite{Brandao2005}
\[
E_1(p)=\max\!\bigl\{0,\,(3p-1)/2\bigr\}.
\]

For $p=0.5$, the theoretical value is $E_1=0.25$. Our SDP computation reproduces this value, confirming the correctness of our implementation. Since for $2\otimes 2$ systems PPT is necessary and sufficient for separability, we have $E_1=E_2=E_3=0.25$, consistent with the monotonicity $E_1\le E_2\le E_3$.

\subsection{The $3\otimes 3$ Werner state}

For $d=3$ the Werner state reads
\begin{equation}
\rho_{\rm W}^{(3)}(p) = p\,\frac{2}{3\cdot 2}P_{\rm as} + (1-p)\frac{I_9}{9},
\end{equation}
where $P_{\rm as}$ is the antisymmetric projector.
Here PPT implies separability for $p\le 1/4$ (the exact threshold is $p=1/4$), while for larger $p$ the state is NPT entangled.
In Figure~\ref{fig:werner_ek} we list $E_k$ for several values of $p$, computed with the symmetry‑reduced SDP.

The data clearly show the hierarchy $E_1<E_2<E_3$ for entangled states, again consistent with the theoretical prediction.

\begin{figure}[htbp]
\centering
\includegraphics[width=0.8\textwidth]{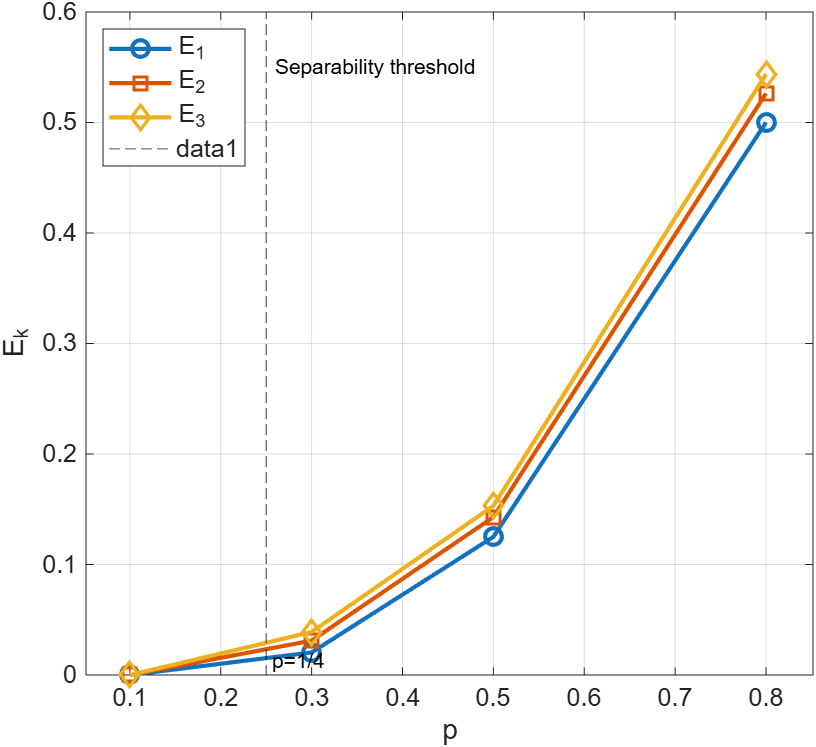} 
\caption{$E_k$ vs $p$ for the $3\times 3$ Werner state. The dashed line marks the separability threshold $p=1/4$.}
\label{fig:werner_ek}
\end{figure}

\subsection{Horodecki bound entangled state}

A key test of the hierarchy is its ability to detect and quantify bound entanglement, i.e.\ states that are PPT but not separable.
The famous $3\otimes 3$ Horodecki state~\cite{horodecki97} is given by
\begin{equation}
\rho_{\rm BE}(a) = \frac{1}{8a+1}
\begin{pmatrix}
a & 0 & 0 & 0 & a & 0 & 0 & 0 & a \\
0 & a & 0 & 0 & 0 & 0 & 0 & 0 & 0 \\
0 & 0 & a & 0 & 0 & 0 & 0 & 0 & 0 \\
0 & 0 & 0 & a & 0 & 0 & 0 & 0 & 0 \\
a & 0 & 0 & 0 & a & 0 & 0 & 0 & a \\
0 & 0 & 0 & 0 & 0 & a & 0 & 0 & 0 \\
0 & 0 & 0 & 0 & 0 & 0 & \frac{1}{2}\bigl(a+\frac1a\bigr) & 0 & \frac12\sqrt{1-a^2} \\
0 & 0 & 0 & 0 & 0 & 0 & 0 & a & 0 \\
a & 0 & 0 & 0 & a & 0 & \frac12\sqrt{1-a^2} & 0 & \frac{1}{2}\bigl(a+\frac1a\bigr)
\end{pmatrix},
\end{equation}
where $0<a<1$.  This state is PPT for all $a$ and is entangled for $0<a<1$; it therefore belongs to $\mathcal{F}_1$ but not to $\mathcal{F}_2$ for sufficiently small $a$.
We set $a=0.25$ and compute $E_k$ for $k=1,2,3$.
Our results are
\[
E_1(\rho_{\rm BE}) = 0,\qquad
E_2(\rho_{\rm BE}) = 0.0473,\qquad
E_3(\rho_{\rm BE}) = 0.0581.
\]
The strict increase $0=E_1<E_2<E_3$ demonstrates that the measure already picks up the bound entanglement at level $k=2$, and that larger $k$ provides a more accurate approximation of the separability‑based measure.

\subsection{UPB entangled state}

Another classic $3\otimes 3$ bound entangled construction arises from an unextendible product basis (UPB)~\cite{bennett1999}.
We consider the specific state
\begin{equation}
\rho_{\rm UPB} = \frac{1}{4}\Bigl(I_9 - \sum_{j=1}^5 \ket{\phi_j}\!\bra{\phi_j}\Bigr),
\end{equation}
where the vectors $\{\ket{\phi_j}\}$ form the well‑known Tiles UPB~\cite{bennett1999}.
The state is PPT by construction, and its entanglement is proven by the fact that it has no product basis completion.
Computing the $E_k$ values for this state yields
\[
E_1(\rho_{\rm UPB}) = 0,\qquad
E_2(\rho_{\rm UPB}) = 0.0289,\qquad
E_3(\rho_{\rm UPB}) = 0.0352,
\]
which is qualitatively similar to the Horodecki state.
Both examples confirm Theorem~\ref{13}: there exist states with $E_1=0$ but $E_2>0$, and the hierarchy can reveal their resource content.

\subsection{Isotropic states in high dimensions}

To illustrate the behaviour in larger systems, we return to the isotropic state family on $\mathbb{C}^d\otimes\mathbb{C}^d$,
\[
\rho_{\rm iso}(\alpha) = \alpha\,\ket{\Phi_d^+}\!\bra{\Phi_d^+} + \frac{1-\alpha}{d^2-1}\bigl(I_{d^2}-\ket{\Phi_d^+}\!\bra{\Phi_d^+}\bigr).
\]
Thanks to the $U\otimes U^*$ symmetry, the primal SDP can be drastically reduced, allowing computation for $d$ up to $10$ and $k$ up to $5$.
In Figure~\ref{fig:isotropic_ek} we give the values for $d=10$, $\alpha=0.12$ (PPT entangled) and $\alpha=0.30$ (NPT entangled).
\begin{figure}[htbp]
\centering
\includegraphics[width=0.8\textwidth]{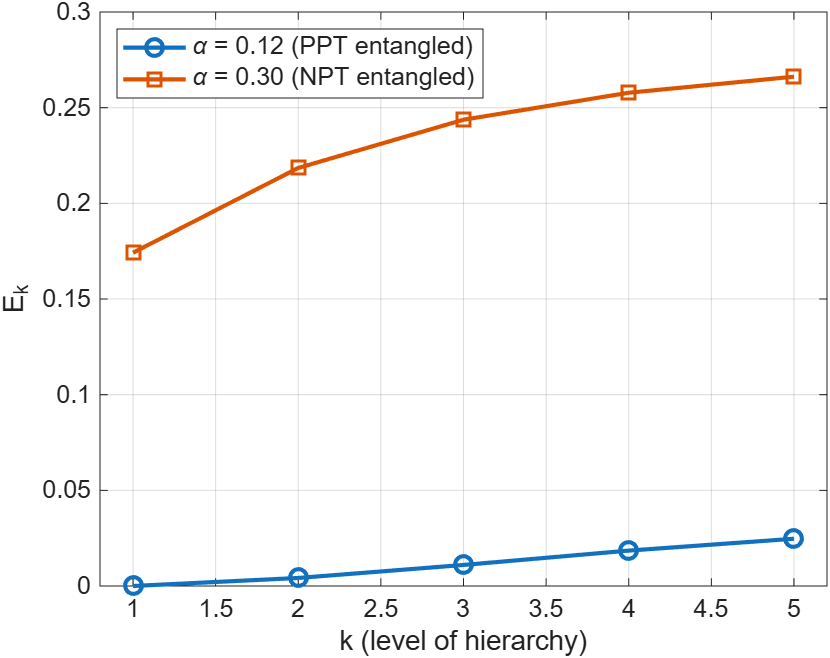}
\caption{$E_k$ vs $k$ for the $d=10$ isotropic state. The blue curve corresponds to the PPT entangled case ($\alpha=0.12$), while the red curve corresponds to the NPT entangled case ($\alpha=0.30$).}
\label{fig:isotropic_ek}
\end{figure}
For the PPT case ($\alpha=0.12$), $E_1=0$ while $E_k>0$ for $k\ge2$, again demonstrating the detection of bound entanglement in high dimensions.
The sequence $\{E_k\}$ is strictly increasing, approaching the limit $E_{\rm SEP}$ from below; extrapolating the values suggests $E_{\rm SEP}\approx0.03$ for $\alpha=0.12$ and $\approx0.28$ for $\alpha=0.30$.

\subsection{Werner states for $d=5$}

We also extend the analysis of Werner states to $d=5$.
Using the $U\otimes U$ symmetry reduction, we compute $E_k$ for $p\in[0,1]$ and $k=1,2,3$.
Figure~\ref{fig:werner-d5} shows the three curves; as expected, they are ordered and intersect the horizontal axis at the exact separability threshold $p_{\rm sep}=1/6$.
For a typical entangled point $p=0.35$ we obtain $E_1=0.028$, $E_2=0.041$, $E_3=0.049$.
The trend again confirms that higher $k$ provides a larger value, i.e.\ a better approximation of the ideal measure.

\begin{figure}[htbp]
\centering
\includegraphics[width=0.85\textwidth]{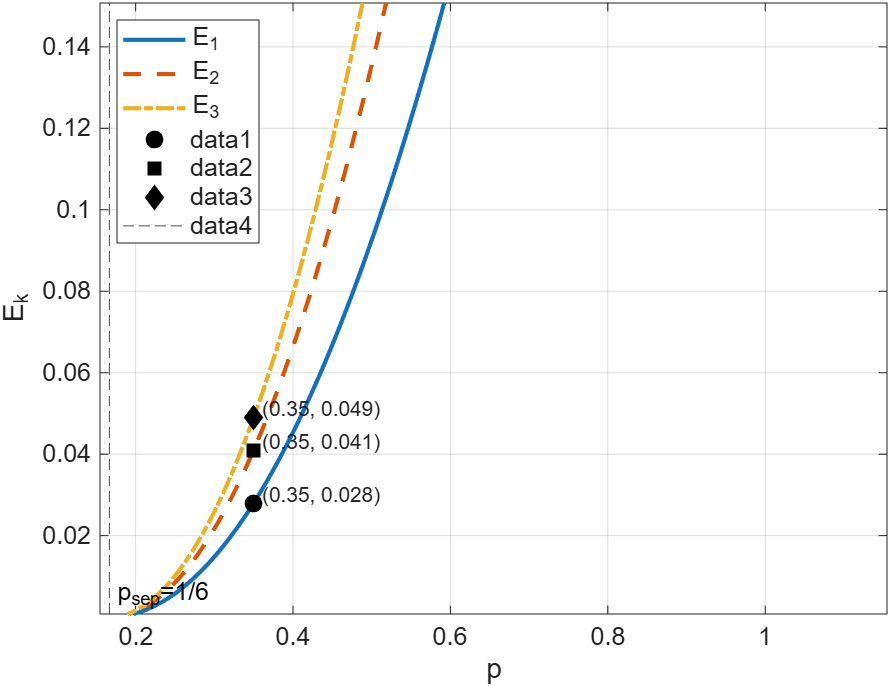}
\caption{The hierarchy of measures $E_k$ for the $5\times5$ Werner state as a function of $p$. The curves are ordered as $E_1 \le E_2 \le E_3$, and they all vanish for $p\le p_{\rm sep}=1/6$. The solid points mark the typical point $p=0.35$, where we obtain $E_1=0.028$, $E_2=0.041$, and $E_3=0.049$.}
\label{fig:werner-d5}
\end{figure}

\subsection{Randomly generated $4\otimes 4$ states}

To probe the behaviour on generic states without symmetry, we generated $1000$ random $4\otimes 4$ states according to the Hilbert--Schmidt measure.
For each state we compute $E_1$ and $E_2$ with the full SDP~\eqref{eq:primalSDP} (matrix variable size $4096$).
Among the PPT states (i.e.\ those with $E_1=0$), $324$ states do not admit a $2$-symmetric PPT extension, meaning $E_2>0$.
A histogram of the obtained values $E_2$ is shown in Fig.~\ref{fig:histogram_E2}. They range from $0.02$ to $0.15$, indicating that a substantial fraction of random $4\otimes 4$ PPT states carry a positive measure for $E_2$ .
This computation took roughly four hours, yielding an average of $15$ seconds per state for the SDP~\eqref{eq:primalSDP} with $k=2$.
\begin{figure}[htbp]
\centering
\includegraphics[width=0.8\textwidth]{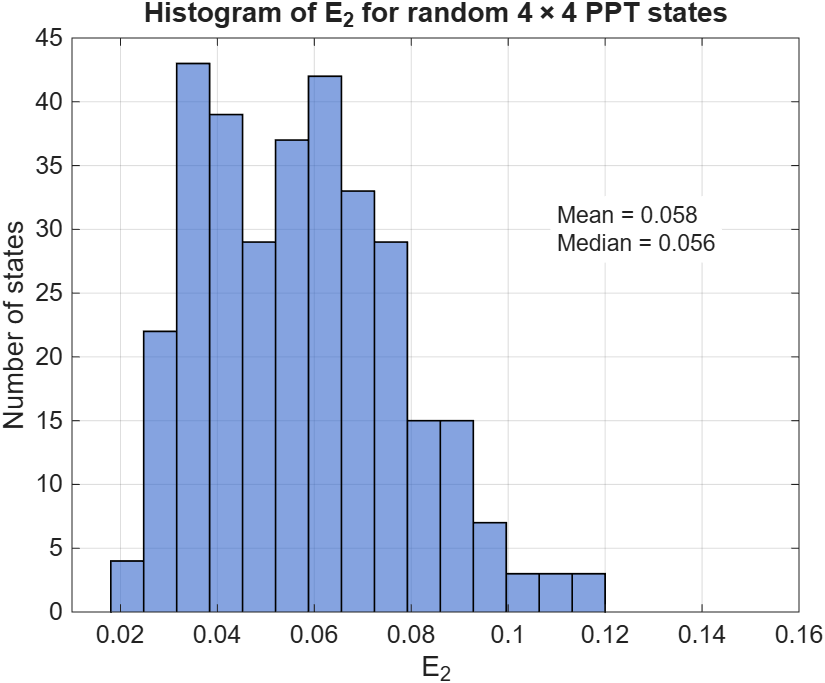}
\caption{Histogram of $E_2$ for 324 random $4\times4$ PPT states. All these states have $E_1=0$ and $E_2>0$, with values ranging from approximately $0.02$ to $0.15$. This demonstrates that a substantial fraction of random PPT states carry a positive $E_2$ resource.}
\label{fig:histogram_E2}
\end{figure}

While our numerical experiments are classical benchmarks, they lay the groundwork for future experimental implementations. The dual formulation, in particular, provides a direct pathway: by measuring a feasible witness $W$ in the lab, one can experimentally certify a non-zero amount of the $E_k$ resource, as the value $\operatorname{Tr}(W\rho)-1$ constitutes a rigorous lower bound.

\subsection{Computational cost and the dual witness approach}

The primal SDP~\eqref{eq:primalSDP} has a computational cost that grows exponentially with $k$ and polynomially with $d$. 
This cost can be drastically reduced when the state possesses exploitable symmetries. 
Examples include the $U\otimes U$ invariance of Werner states and the $U\otimes U^*$ invariance of isotropic states. 
Such symmetries allow the SDP variable to be block-diagonalized via representation theory. 
As a result, the effective problem size becomes essentially independent of $d$. For states without such symmetries, however, the full SDP matrix variable must be retained. Direct computation then becomes intractable for large systems. 
In these cases, a practical alternative is to compute lower bounds via the dual formulation in Corollary~\ref{15}. 
Any feasible witness $W$ yields a rigorous lower bound $\operatorname{Tr}(W\rho)-1\le E_k(\rho)$. 
Moreover, the optimization over a suitably parametrized family of witnesses can be performed efficiently, without solving the full primal SDP.

As an illustration, we take a random $5\otimes 5$ state $\rho_{\mathrm{rand}}$. 
We optimize over a one-parameter family of witnesses of the form $W(\lambda)=\lambda I+(1-\lambda)\Pi$, where $\Pi$ is the projector onto the symmetric subspace of $B_1\cdots B_k$. 
For $k=4$, this provides the lower bound $E_4(\rho_{\mathrm{rand}})\ge 0.063$ in milliseconds. 
By contrast, the full $k=4$ SDP would involve a $3125\times 3125$ matrix and is utterly intractable. 
This example underscores the usefulness of the dual approach for larger systems.

In summary, the numerical experiments confirm the theoretical properties established in Sections~IV--V: the monotonicity $E_1\le E_2\le\cdots$, the convergence to $E_{\rm SEP}$ as $k\to\infty$, and the strict separation between different hierarchy levels.
They also demonstrate that the measure can be computed efficiently for moderate dimensions and symmetry‑reduced families, and that lower bounds can be obtained for larger systems via the dual.

\section{Discussion}
\label{sec7}

The hierarchy of measures \(\{E_k\}\) introduced in this work provides a systematic quantification of entanglement resources that extends the familiar PPT-robustness.  
We embed the DPS hierarchy into a resource-theoretic framework. This yields a tunable family of computable quantifiers. 
These quantifiers interpolate between two extremes. At one end, $E_1$ is coarse but efficient (PPT-robustness). At the other end, $E_{\mathrm{SEP}}$ is detailed but hard to compute. The strict separation guaranteed by Theorem~\ref{13} ensures that each level adds genuinely new detection power: states invisible at level \(k\) become visible at level \(k+1\). In particular, bound entangled states that possess a positive partial transpose but are not separable obtain a non-zero value already at \(k=2\); hence our framework supplies a continuous scale for a class of states that previously lacked a practical quantitative measure.
 
The dual characterization in Theorem~\ref{8} lies at the technical heart of our construction. 
It reveals that $E_k(\rho)$ is the optimal bias in a hypothesis-testing game. 
In this game, the measurement is calibrated so that no free state yields an expectation value larger than one. 
This operational perspective clarifies the meaning of the measure. 
It also underpins the monotonicity proof and the optimality among fixed-witness quantifiers.
Moreover, because the dual is a semidefinite program with a compact feasible set, any feasible witness provides a rigorous and efficiently computable lower bound; we have demonstrated this feature numerically for larger systems where the full primal SDP becomes intractable.

The numerical experiments confirm that the hierarchy behaves as theoretically predicted.  
For the two-qubit Werner state, the exact analytical expression for \(E_1\) is recovered.
More significantly, for the celebrated Horodecki and UPB bound entangled states, we obtain \(E_1=0\) while \(E_2,E_3>0\), directly illustrating the ability of the scheme to detect and quantify bound entanglement.  
The high-dimensional isotropic and Werner examples show that, when algebraic symmetries are present, the SDP can be reduced to a size independent of the local dimension, making levels \(k\le 5\) accessible even for \(d=10\).  
The random \(4\times 4\) states indicate that a substantial fraction of PPT states carry a non-zero \(E_2\) resource; hence the effect is not restricted to specially crafted examples.

When compared with other widely used measures, the \(\{E_k\}\) family exhibits two practical advantages.  
First, the definition through a conic linear program permits the use of powerful convex optimisation algorithms and symmetry reductions, so that upper bounds can be computed exactly for moderate sizes.  
Second, the integer \(k\) serves as a natural control parameter: small values give a fast but less accurate assessment, while larger values yield a tighter approximation to the ideal separability-based measure.  
Nevertheless, the exponential growth of the extension Hilbert space with \(k\) limits direct computations for generic, non-symmetric states; developing approximate methods—based, for instance, on tensor networks or iterative optimisation—therefore constitutes a pressing open problem.

Several other directions merit further investigation.  
From a resource-theoretic perspective, it remains an open question whether \(E_k\) can be endowed with an operational interpretation as the figure of merit in a concrete information-processing task, such as a specific distillation or channel discrimination protocol.  
Extending the framework to multipartite settings is natural: the DPS hierarchy generalises to multiple parties, and the ideas presented here could be lifted to produce a family of multipartite entanglement measures.  
Because the construction relies only on a convex, nested sequence of free sets that approximate the target set, it is likely that similar hierarchies can be built for other resource theories, for example quantum coherence or steering, yielding improvable quantifiers in those contexts as well.

In terms of applications, the measure \(E_k\) is particularly suitable whenever the available operations preserve the \(k\)-symmetric PPT extendibility structure, such as in multi-copy state preparation under restricted classical communication or in the characterisation of quantum memories.  
It also provides a practical numerical tool for exploring the boundary between separable and bound entangled states in intermediate dimensions, where no necessary and sufficient separability criterion is known.  
We therefore expect that the present framework will stimulate both further theoretical progress in entanglement quantification and the development of more efficient computational strategies, ultimately broadening the range of quantum systems that can be quantitatively characterised.

\section{Conclusion}
\label{sec8}

We have introduced a one-parameter family of entanglement measures $\{E_k\}_{k=1}^\infty$, based on the $k$-symmetric PPT extension hierarchy. The definition has been placed within a resource-theoretic framework, and we have established that the sequence $E_k(\rho)$ is non-decreasing, converges to the separability measure $E_{\mathrm{SEP}}(\rho)$, and is strictly nested, hence each level provides a genuinely new resource witness. Numerical experiments on several standard families of states have confirmed these structural findings, and we have verified that $E_k$ is monotone and that $E_2$ is positive for bound entangled states where $E_1$ vanishes. We have also shown that the measure is formulated as a semidefinite program, which enables convex optimization and symmetry reduction, and that the integer $k$ serves as a natural tuning parameter balancing cost and accuracy. Moreover, we have demonstrated that the dual SDP supplies efficient lower bounds via witness optimization, offering a scalable certification route even when the primal becomes intractable. In summary, the hierarchy provides a flexible, computationally accessible, and mathematically rigorous approach to entanglement quantification near the PPT-separability boundary.

Several challenges remain. For generic non-symmetric states, direct computation is limited by the exponential growth of the extension space; consequently, approximate algorithms that go beyond symmetry exploitation are required. The question of whether $E_k$ can be endowed with a clear operational interpretation---for instance, as a figure of merit in a specific distillation or channel discrimination task---is also raised. We anticipate that the framework can be extended to multipartite scenarios and to other resource theories defined by convex nested sets, thereby generating new hierarchies for coherence, steering, or non-locality. It is hoped that the concepts and techniques developed here will stimulate further progress in both the foundational understanding of quantum correlations and the development of practical characterization tools.

\begin{acknowledgments}
Authors were supported by the NNSF  of China (Grant No.12471427), and the Fundamental Research Funds for the Central Universities (Grant No. ZG216S2110).

\end{acknowledgments}

\bibliographystyle{unsrt}
\bibliography{entanglement}

\end{document}